\documentclass[journal]{IEEEtran}

\usepackage{cite,bm,tikz,amssymb,amsthm,amsmath,graphicx,pgfplots,multirow,epstopdf,algorithm,algorithmicx,algpseudocode,url}
\usepackage[acronym]{glossaries}
\usetikzlibrary{arrows.meta,decorations.pathreplacing,plotmarks,shapes,arrows,positioning}
\pgfplotsset{major grid style={solid,white!90!black}}
\pgfplotsset{every tick label/.append style={font=\scriptsize}}
\pgfplotsset{tickwidth=0.1cm}
\let\originalleft\left
\let\originalright\right
\renewcommand{\left}{\mathopen{}\mathclose\bgroup\originalleft}
\renewcommand{\right}{\aftergroup\egroup\originalright}

\definecolor{mycolor1}{rgb}{0.00000,0.44700,0.74100}%
\definecolor{mycolor2}{rgb}{0.85000,0.32500,0.09800}%
\definecolor{mycolor3}{rgb}{0.92900,0.69400,0.12500}%
\definecolor{mycolor4}{rgb}{0.49400,0.18400,0.55600}%
\definecolor{mycolor5}{rgb}{0.46600,0.67400,0.18800}%
\definecolor{mycolor6}{rgb}{0.30100,0.74500,0.93300}%

\def\Nch{N}
\def\rpnvar{\sigma^2_\mathrm{p}}
\renewcommand\vec{\underline}
\def\rnd{\bm}
\newcommand{\rndvec}[1]{\rnd{\vec{#1}}}
\def\rndmtx{\mathbf}

\newtheorem{prop}{Proposition}

\DeclareMathOperator*{\argmax}{argmax}
\DeclareMathOperator*{\argmin}{argmin}
\DeclareMathOperator*{\diag}{diag}

\algnewcommand\algorithmicforeach{\textbf{for each}}
\algdef{S}[FOR]{ForEach}[1]{\algorithmicforeach\ #1\ \algorithmicdo}

\usepgfplotslibrary{external} 
\begin{document}

\title{Optimization~of~Transmitter-Side~Signal\\Rotations~in~the~Presence~of~Laser~Phase~Noise}

\author{
	Arni~F.~Alfredsson,~\IEEEmembership{Student~Member,~IEEE},~Erik~Agrell,~\IEEEmembership{Fellow,~IEEE},\\Magnus~Karlsson,~\IEEEmembership{Senior~Member,~IEEE,~Fellow,~OSA},~and~Henk~Wymeersch,~\IEEEmembership{Member,~IEEE}%
\thanks{
A. F. Alfredsson, E. Agrell, and H. Wymeersch are with the Department of Electrical Engineering, Chalmers University of Technology, SE-41296 G\"{o}teborg, Sweden (e-mail: arnia@chalmers.se).

M. Karlsson is with the Photonics Laboratory, Department of Microtechnology and Nanoscience, Chalmers University of Technology, SE-41296 G\"{o}teborg, Sweden.

This work was supported by the Swedish Research Council (VR) via Grants 2014-6138 and 2018-03701 and the Knut and Alice Wallenberg Foundation via Grant 2018.0090.
}}
\maketitle

\begin{abstract}
	\boldmath
	The effects of transmitter-side multidimensional signal rotations on the performance of multichannel optical transmission are studied in the presence of laser phase noise.
	In particular, the laser phase noise is assumed to be uncorrelated between channels.
	To carry out this study, a simple multichannel laser-phase-noise model that has been experimentally validated for weakly-coupled multicore-fiber transmission is considered.
	As the considered rotation scheme is intended to work in conjunction with receiver-side carrier phase estimation (CPE), the model is modified to further assume that imperfect CPE has taken place, leaving residual phase noise in the processed signal.
	Based on this model, two receiver structures are derived and used to numerically optimize transmitter-side signal rotations through Monte Carlo simulations.
	For reasonable amounts of residual phase noise, rotations based on Hadamard matrices are found to be near-optimal for transmission of four-dimensional signals.
	Furthermore, Hadamard rotations can be performed for any dimension that is a power of two.
	By exploiting this property, an increase of up to 0.25 bit per complex symbol in an achievable information rate is observed for transmission of higher-order constellations.
\end{abstract}


\section{Introduction}
\label{sec:intro}

One of the main limiting factors in fiber-optic systems is phase noise stemming from laser imperfections and the optical Kerr effect in fibers \cite{Ip:08}. Laser phase noise is typically compensated for using receiver-based digital signal processing (DSP) methods such as the blind phase search algorithm \cite{4814758}. Pilot-aided algorithms are also becoming increasingly utilized due to their robustness at low signal-to-noise ratios (SNRs) and transparency to modulation formats \cite{Mazur:19}. Furthermore, specialized phase-noise mitigation schemes such as self-homodyne detection have been demonstrated in space-division multiplexed (SDM) and wavelength-division multiplexed (WDM) transmission \cite{mazur_jlt18,Puttnam:13}. For the compensation of nonlinear phase noise, studies on DSP-based solutions that have proved beneficial include digital backpropagation \cite{1576188,4738549} and Kalman equalization \cite{6888491}.

All practical methods for phase-noise compensation are imperfect in the sense that they do not completely mitigate the laser and nonlinear phase noise. As a result, the residual phase noise may still impair system performance. This impact may be reduced in various ways \cite{6510018,6919303}, in particular by applying DSP on the transmitter side \cite{7478061,5979171}.

Transmitter-side DSP has been extensively studied in wireless communications and has led to, e.g., space--time codes \cite{771146} and channel-aware precoding schemes \cite{1207369} for multiple-input multiple-output (MIMO) systems.
More recently in the context of fiber-optic transmission, schemes involving nonlinearity precompensation \cite{1576188,7478061}, preequalization \cite{1396065,Zhang:14:OE}, and precoding \cite{Rath:17} have been proposed. Space--time coding and power-allocation schemes have also been used to maximize power efficiency through waterfilling solutions \cite{Barros:10}, and to mitigate polarization-dependent loss (PDL) \cite{Awwad:13,Zhu:15,8696686} and mode-dependent loss (MDL) \cite{7307951,Amhoud:15}.

A potential drawback of some transmitter-based schemes is that they often require the knowledge of instantaneous channel-state information (CSI), which is fed back from the receiver to the transmitter. In fiber-optic systems, the channel tends to change faster than the round-trip delay, thus precluding the use of reliable instantaneous CSI at the transmitter \cite{Amhoud:15}. Methods that only use statistical or no CSI are therefore of practical interest.

A class of methods that do not require CSI involves the rotation of multidimensional signals. Rotations have been used in single-antenna wireless communications to improve diversity performance for transmission through a fading channel \cite{681321}. In MIMO transmission, rotation-based space--time codes have been proposed \cite{985979}. CSI-independent rotation schemes have also been considered in fiber-optic communications \cite{6718045}. In particular, rotations based on Hadamard matrices have been used to mitigate PDL \cite{6647786,Awwad:13,8696686}, MDL \cite{7537474}, and channel-dependent loss due to imperfect filtering in WDM transmission \cite{Shibahara:15}. However, this concept has not been explored for fiber-optic transmission impaired by laser phase noise.

Inspired by the presence of residual phase noise due to nonideal DSP, we pose the following questions: Do transmitter-side multidimensional signal rotations change the impact of residual phase noise on the transmission performance, and if so, which rotations are effective?
To address these questions, we study the effect of transmitter-side signal rotations on the performance of multichannel systems in the presence of phase noise.
We consider scenarios where the Kerr nonlinearity is sufficiently weak so that the laser phase noise will be the dominant source of the residual phase noise after compensation.
Moreover, as it is already well known that correlated laser phase noise can be exploited to improve performance, we focus on the case where the laser phase noise is uncorrelated across the channels.

A simple multichannel laser-phase-noise model that has been experimentally validated for transmission through a weakly-coupled, homogeneous, single-mode multicore fiber (MCF) is considered. This system model is modified to further assume the use of nonideal receiver-side carrier-phase estimation (CPE), leaving residual phase noise. Based on the modified model, two receiver structures are derived that perform joint-channel and per-channel symbol detection, thus offering different complexity and performance. Transmission of rotated constellations is considered, and this rotation is numerically optimized for the derived receivers with respect to different performance metrics for a number of model parameters.

The contributions of the paper comprise the following results:
1) On one hand, signal rotations can improve performance in the presence of moderate laser phase noise, even when the phase-noise statistics are not exploited. On the other hand, rotations are not beneficial in the case of extreme amounts of laser phase noise;
2) Consider two-channel transmission. Rotations based on Hadamard matrices are near-optimal in terms of block error rate (BLER), bit error rate (BER), and achievable information rate (AIR) for moderate amounts of laser phase noise;
3) Assume transmission of Hadamard-rotated constellations over a large number of complex channels. Derotating the received signal before data detection causes the residual phase noise to manifest as attenuation and additive noise. This can improve an achievable information rate by up to 0.25 bit per complex symbol for higher-order modulations.

\textit{Notation:} Vectors are denoted by underlined letters $\vec x$, whereas matrices are expressed by uppercase sans-serif letters $\mathsf X$, and $\mathsf X_N$ is used when the size of a square matrix needs to be stated explicitly. Sets are indicated by calligraphic letters $\mathcal X$. Boldface letters denote random quantities. The real line and complex plane are represented by $\mathbb R$ and $\mathbb C$, the real and imaginary parts of a complex number are denoted by $\Re$ and $\Im$, and $j$ represents the imaginary unit. The probability mass function (PMF) of a discrete random variable $\rnd x$ at $x$ is written as $P_{\rnd{x}}(x)$, and the probability density function (PDF) of a continuous random variable $\rnd x$ at $x$ is denoted by $p_{\rnd{x}}(x)$. The probability distribution of a mixed discrete--continuous random variable is expressed in the same way as PDFs. The expectation of a random variable is indicated by $\mathbb{E}[\cdot]$.

\section{System Model}
\label{sec:sysmodel}

The considered system model is based on a simple multichannel laser-phase-noise model proposed in \cite{8576586}, which was used to assess potential performance gains that come from estimating laser phase noise jointly over multiple channels. The model-based simulation results were found to be in agreement with experimental results for transmission through a weakly-coupled, single-mode, homogeneous MCF. The system model in \cite{8576586} entails uncoded dual-polarization transmission over multiple channels in the presence of laser phase noise, which is arbitrarily correlated among the channels. The data symbols are modeled as uniform random variables that take on values from a set $\mathcal X$ of complex constellation points, and normalized such that the average symbol energy is $E_\mathrm{s}$ per channel. The inclusion of pilots within the transmitted symbol block is also assumed, whose values and positions are known to the transmitter and receiver. Moreover, the fiber Kerr nonlinearity is neglected and the received signal is assumed to have undergone ideal DSP such that all impairments except for laser phase noise have been fully mitigated. 

In this work, we use the same assumptions as the model in \cite{8576586} with an additional constraint: The laser phase noise is assumed to be uncorrelated between channels. This is a reasonable assumption when independent lasers are used for each channel and in certain cases when lasers are shared among channels\footnote{In SDM MCF transmission, intercore skew causes the signals to reach the receiver at different times, which decorrelates the phase noise across the channels \cite{7328945}. In WDM frequency-comb transmission, chromatic dispersion has an analogous effect \cite{7390165}.}. Furthermore, the transmission of arbitrarily rotated multidimensional constellations over multiple channels (wavelength or space) is considered. The resulting complex constellations that are transmitted in each channel can be regarded as two-dimensional (2D) projections of the rotated multidimensional constellation. For an arbitrary rotation, the complex constellations will lack structure, which can complicate or prevent the implementation of blind DSP methods. Therefore, all receiver DSP blocks are assumed to be implemented in a pilot-aided manner, operating independently of the data modulation, as experimentally demonstrated in \cite{Mazur:19}.

Nonideal pilot-aided CPE is assumed to have been performed in the receiver, leaving residual phase noise in the processed signal. In general, the residual phase noise is not perfectly memoryless over time since nonideal algorithms are unable to completely remove the memory from the laser phase noise. Moreover, the residual phase noise is nonstationary since the variance of each sample depends on the temporal distance to its nearest pilot, albeit the mean of all samples is zero. This is exemplified in Fig.~\ref{fig:PEdist},
\begin{figure}[!t]
	\centering
	\includegraphics{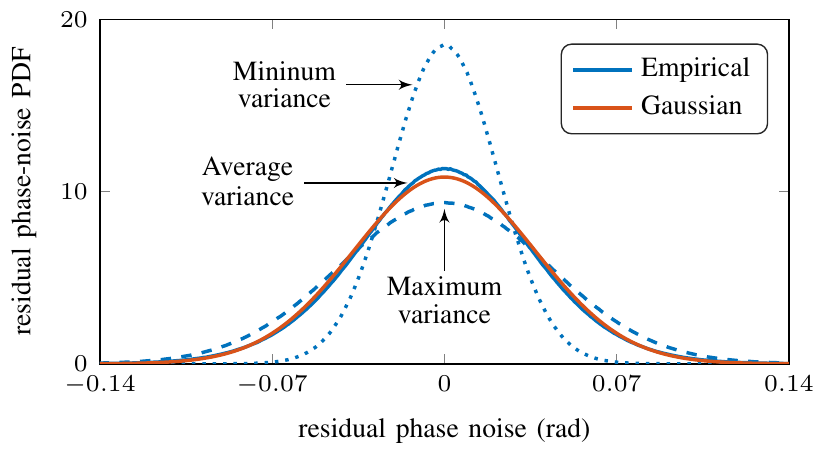}
	\vspace{-0.1cm}
	\caption{Empirical distribution and fitted Gaussian PDF of the residual phase noise for transmission at 20 GBaud with a combined laser linewidth of 200 kHz at an SNR of 30 dB.}
	\label{fig:PEdist}
\end{figure}
where the empirical distribution is plotted for samples of the residual phase noise that have the minimum and maximum variance. In addition, the average distribution of the residual phase noise across the whole time span of the transmitted frame is also plotted. The curves are obtained for a typical set of system parameters when the well-known iterative CPE algorithm proposed by Colavolpe \textit{et al.} in \cite[Sec.~IV-B]{1504908} is run for a single iteration with 1\% pilot rate\footnote{In coded transmission, the algorithm also uses soft estimates of the data symbols through iterations with the decoder. During the initial iteration, however, it operates in pilot-aided mode where it does not make use of the data symbols.}.

The difference in the variance of residual phase-noise samples is highly dependent on the SNR, the laser linewidth, the baud rate, and the positions of pilot symbols. In particular, the difference becomes negligible at lower SNRs, assuming practical laser linewidths and baud rates. Hence, for simplicity, the residual phase noise in each channel is assumed to be a memoryless and stationary zero-mean process with variance corresponding to the average distribution shown in Fig.~\ref{fig:PEdist}.
\begin{figure}[!t]
	\centering
	\includegraphics{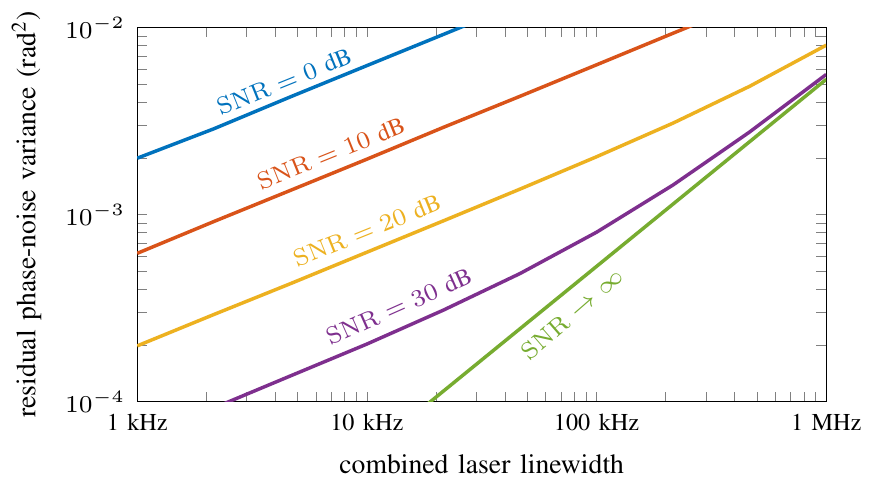}
	\vspace{-0.5cm}
	\caption{The estimated variance of the residual phase noise versus the combined laser linewidth of the system.}
	\label{fig:typicalRPNvar}
\end{figure}
This distribution can be approximated as a zero-mean Gaussian PDF, whose variance $\rpnvar$ is estimated using the residual phase noise across the whole transmitted frame. Assuming residual phase noise to be either Tikhonov or Gaussian distributed is an established approach in wireless communications \cite{6510018}.

The residual phase noise is assumed to be uncorrelated across channels. This is a valid assumption when the laser phase noise is uncorrelated across channels. However, as the two polarizations of optical fields are generated by a single laser, the residual phase noise will in general be highly correlated among the polarizations of each channel.
Hence, we consider a scheme where a multidimensional signal rotation is applied to the x-polarizations in all channels, and a separate but identical rotation is applied to all y-polarizations. This scheme enables us to assume the residual phase noise to be uncorrelated in all complex dimensions in which the signal rotation is performed.

Fig.~\ref{fig:typicalRPNvar} shows the estimated variance of the residual phase noise as a function of the combined light-source and local oscillator laser linewidth for 20 GBaud transmission at different SNRs. For SNRs larger than 10 dB, typical variance values are in the range $[10^{-4},10^{-2}]$ rad\textsuperscript2. As a reference, an SNR of 10 dB yields a BER of roughly $6\cdot10^{-2}$ for Gray-coded 16QAM transmission over the complex additive white Gaussian noise (AWGN) channel.

Given the above assumptions, the system model thus entails the propagation of a multidimensional constellation through $\Nch$ channels, in addition to fixed receiver DSP. This leaves residual phase noise which is memoryless and stationary, as well as uncorrelated across the channels. The signals are considered to be impacted identically in all channels. Assuming one complex sample per symbol, the resulting processed signal corresponding to a transmitted vector of signals $\tilde{\rndvec{s}}\in\mathbb C^\Nch$ across all channels is $\rndvec{r}=\rndmtx\Theta\tilde{\rndvec{s}}+\rndvec{n}$, where $\rndmtx\Theta=\diag([e^{j\rnd\theta_1},\dots,e^{j\rnd\theta_\Nch}])$ is a diagonal matrix.
Denoting the vector transpose by $(\cdot)^T$, the residual phase noise is modelled as a Gaussian random vector, $\rnd{\vec\theta}=[\rnd\theta_1,\dots,\rnd\theta_\Nch]^T$, whose elements are assumed to be jointly Gaussian distributed. Due to the assumption of uncorrelated residual phase noise, $\rndvec{\theta}$
contains $\Nch$ independent and identically distributed (i.i.d.) zero-mean Gaussian random variables with variance $\rpnvar$. Finally, amplified spontaneous emission is accounted for and modelled as complex AWGN $\rndvec{n}=[\rnd n_{1},\dots,\rnd n_{\Nch}]^T$, comprising $\Nch$ i.i.d. circularly symmetric complex Gaussian random variables with variance $N_0$.

The transmitter optimization involves finding an effective rotation that maps the data symbols $\rndvec{s}\in\mathcal X^\Nch$ into transmitted signals $\rndvec{\tilde s}\in\mathcal C^\Nch$, where the multidimensional rotation is performed on a real-component basis. Such a rotation does not affect the performance over the AWGN channel, but as we will show, it can be beneficial for channels with phase noise.
Let $g([z_1,\dots,z_\Nch]^T)=[\Re\{z_1\},\Im\{z_1\},\dots,\Re\{z_\Nch\},\Im\{z_\Nch\}]^T$ and $g^{-1}([x_1,\dots,x_{2\Nch}]^T)=[x_1+jx_2,\dots,x_{2\Nch-1}+jx_{2\Nch}]^T$ for $z_1,\dots,z_\Nch\in\mathbb C$ and $x_1,\dots,x_{2\Nch}\in\mathbb R$. Then, $\rndvec{\tilde{s}}=f_\mathsf{R}(\rndvec{s})=g^{-1}(\mathsf Rg(\rndvec{s}))$, and the received and processed signal can be described as
\begin{equation}
	\rndvec{r}=\rndmtx\Theta f_{\mathsf R}(\rndvec{s})+\rndvec{n},
	\label{eq:sysmodel}
\end{equation}
for $\mathsf R\in\mathcal O^+$, where $\mathcal O^+$ denotes the set of orthogonal matrices with determinant $+1$, i.e., rotation matrices. Note that since $\mathsf R^T \mathsf R=\mathsf I$, $f_{\mathsf R}^{-1}=f_{\mathsf R^T}$ because $f_{\mathsf R^T}(f_{\mathsf R}(\rndvec{s}))=f_{\mathsf R^T\mathsf R}(\rndvec{s})=\rndvec{s}$. Furthermore, that the components of $\rndvec{\tilde s}$ are statistically dependent for a general $\mathsf R$, even though the components of $\rndvec{s}$ are assumed to be independent of each other.

\section{Proposed Receivers}
\label{sec:resLPNhad}

In this section, two receivers are considered for the model in \eqref{eq:sysmodel}. The first receiver performs joint-channel symbol detection using knowledge about the phase noise distribution, $p_{\rndvec{\theta}}(\vec\theta)$, and has high performance for the considered system model. The second receiver operates without the knowledge of $p_{\rndvec{\theta}}(\vec\theta)$, and hence, this receiver performs per-channel detection and has lower complexity compared to the first one, at the cost of detection performance.

\subsection{Joint-Channel Receiver Exploiting Phase-Noise Statistics}
\label{sec:jc_rec}

In order to derive a high-performance receiver that depends on the knowledge of $p_{\rndvec{\theta}}(\vec\theta)$, where $\vec\theta=[\theta_1,\dots,\theta_\Nch]^T$, it is suitable to use a strategy that minimizes the resulting BLER, namely the maximum \textit{a posteriori} (MAP) detector. Although optimality is guaranteed only with respect to BLER, this strategy is also effective in terms of other performance metrics. It operates according to
\begin{equation}
	\hat{\vec s}=\argmax_{\vec s\in\mathcal X^\Nch} P_{\rndvec{\tilde s}|\rndvec{r}}(\vec{\tilde s}=f_{\mathsf R}(\vec s)|\vec r),
	\label{eq:map}
\end{equation}
where $\mathcal X^{\Nch}$ is the $\Nch$-ary Cartesian power of $\mathcal X$, and $P_{\rndvec{\tilde s}|\rndvec{r}}(\vec{\tilde s}|\vec r)$ is the \textit{a posteriori} PMF of $\rndvec{\tilde s}$ at $\vec{\tilde s}$ given $\rndvec{r}=\vec r$. Finding the argument that maximizes this PMF leads to a MAP estimate of $\vec{s}$ since $f_{\mathsf R}$ is bijective.
It is shown in Appendix \ref{app:jointrec} that this detection strategy can be carried out approximately as
\begin{align}
	\hat{\vec s}&=\argmax_{\vec s\in\mathcal X^\Nch}\sum_{i=1}^\Nch\left[|\eta_i|-\frac{|\tilde s_{i}|^2}{N_0}-\frac12\log_e|\eta_i|\right],
	\label{eq:log_approx_map}
\end{align}
where
\begin{equation}
	\eta_i\triangleq\frac{2r_{i}\tilde{s}^*_{i}}{N_0}+\frac{1}{\rpnvar}.
	\label{eq:eta_i}
\end{equation}
This receiver exhibits an exponential increase in complexity with the number of channels. This is owing to the joint-channel detection of $\vec s$, entailing a maximization over a set of $|\mathcal X|^\Nch$ constellation points, where $|\mathcal X|$ is the number of constellation points in each channel.

\subsection{Per-Channel Receiver Neglecting Phase-Noise Statistics}
\label{sec:pc_rec}
A suboptimal detection strategy that does not exploit $p_{\rndvec{\theta}}(\vec\theta)$ and scales linearly in complexity with the number of channels can be formulated as follows. Assuming $\rpnvar=0$ yields $\rnd\theta_0=\dots=\rnd\theta_\Nch=0$, which further leads to $\rndmtx\Theta=\mathsf I$, where $\mathsf I$ is the identity matrix, and the received signal can thus be expressed as $\rndvec{r}=f_{\mathsf R}(\rndvec{s})+\rndvec{n}$, i.e., the model reduces to transmission of a rotated signal over a multidimensional complex AWGN channel. Rotating the received signal using $\mathsf R^T$ gives
\begin{equation}
	\rndvec{\tilde r}=f_{\mathsf R^T}(\rndvec{r})=\rndvec{s}+f_{\mathsf R^T}(\rndvec{n}),
	\label{eq:derot}
\end{equation}
which does not change the statistics of $\rndvec{n}$ since the noise power in all channels is assumed to be the same. Finally, since the components of $\rndvec{s}$ are independent of each other, symbol detection can be trivially performed on a per-channel basis as
\begin{equation}
	\hat s_{i}=\argmax_{s_{i}\in\mathcal X}P_{\rnd s_i|\rnd{\tilde r_i}}(s_{i}|\tilde r_{i})=\argmin_{s_{i}\in\mathcal X}|s_{i}-\tilde r_{i}|.
\end{equation}
Hence, given the assumption of $\rpnvar=0$, the receiver structure reduces to a derotation of the received signal and Euclidean-distance minimization on a per-channel basis. A high-level overview of the considered system model, transmitter optimization, and proposed receivers is depicted in Fig.~\ref{fig:sys_model}.
\begin{figure}
	\centering
	\includegraphics{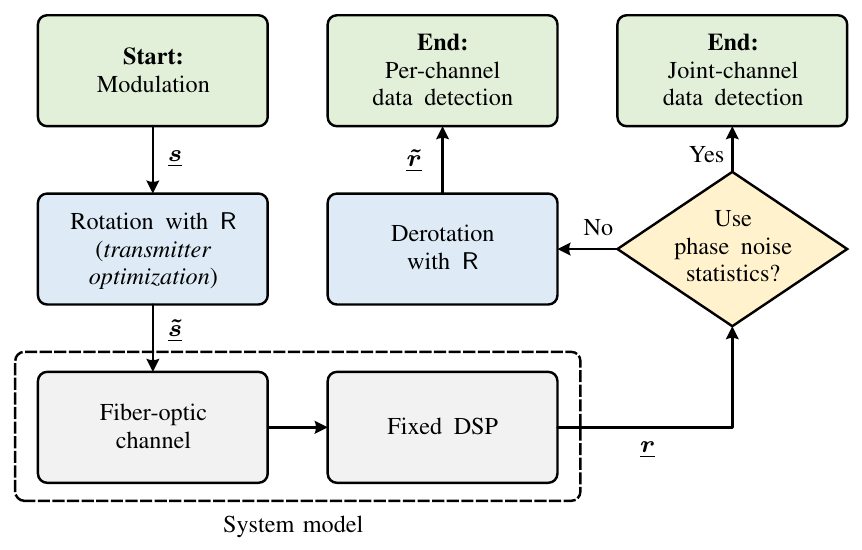}
	\vspace{-0.5cm}
	\caption{A visualization of the considered system model, as well as the strategies for transmitter optimization and data detection.}
	\label{fig:sys_model}
\end{figure}

\section{Rotation-Optimization Results}
\label{sec:rot_opt_results}
This section presents results for the optimization of signal rotations in four dimensions, corresponding to single-polarization transmission of a complex signal through two channels.
Before describing the optimization procedure, we provide a brief overview of Hadamard rotations, which are found to be near-optimal in terms of BLER, BER, and AIR for moderate amounts of residual phase noise.

\subsection{Hadamard Rotations}

Hadamard matrices can be constructed recursively and thus exist for any order $2^\ell$, where $\ell$ is a nonnegative integer\footnote{Various Hadamard matrices whose orders are multiples of 4 can also be constructed, but only orders $2^\ell$ will be considered in this paper.}. Based on this, Hadamard rotations can also be constructed as
\begin{equation}
	\mathsf H_1=1,~
	\mathsf H_2=\frac{1}{\sqrt{2}}
	\begin{bmatrix}
		1&1\\
		-1&1
	\end{bmatrix},~
	\mathsf H_{2^\ell}=\mathsf H_{2}\otimes \mathsf H_{2^{\ell-1}},
	\label{eq:hadamard_construction}
\end{equation}
where $\otimes$ is the Kronecker product and the factor $1/\sqrt{2}$ in the definition of $\mathsf H_2$ ensures that $\mathsf H_{2^\ell}$ is orthogonal for all $\ell$.
Note that the definition of $\mathsf H_2$ is unorthodox in the sense that the columns are swapped with respect to the conventional Walsh--Hadamard matrices.

Since Hadamard matrices consist of only $\pm1$ elements, they are appealing from an implementation standpoint. The multiplication with Hadamard matrices of order $2^\ell$ can be carried out analogously to the fast Fourier transform (FFT), using $2^\ell\log_22^\ell=2^\ell\ell$ additions or subtractions \cite[Ch.~6]{ahmed:orthtr}. Hadamard rotations can be implemented using the same principle, and thus have similar complexity and system-architecture requirements as FFT, which has found applications in, e.g., coherent optical orthogonal frequency-division multiplexed (OFDM) transmission \cite{Shieh:08_2} and CD compensation \cite{Xu:10}. Another characteristic is the potentially simple structure of the transmitted 2D projections after the Hadamard rotation has taken place. As an example, when the constellation points in $\mathcal X$ comprise a subset of a scaled and translated square lattice \cite[Ch.~4]{conway2013sphere}, such as in standard QAM formats, the 2D projections of the rotated signals are also subsets of a scaled and translated square lattice. This puts less stringent resolution requirements on analog-to-digital and digital-to-analog converters compared with arbitrary rotations.

\subsection{Optimization Procedure}

The rotation matrix in \eqref{eq:sysmodel} is numerically optimized with respect to BLER for the joint-channel receiver, and with respect to BER, SER, and AIR in the case of the per-channel receiver, where the AIRs are computed according to \cite[Eq.~(36)]{Alvarado:18}.
Note that $\Nch=2$ is considered and thus, SER and BLER refer to errors in $\mathcal X$ and $\mathcal X^2$, respectively.
Analytical closed-form solutions to these performance metrics are in general difficult to find, and hence, Monte Carlo simulations are used to estimate them. To ensure a reasonable accuracy in the estimation, the performance evaluation of each rotation is based on transmission of at least $10^6$ symbols. Due to the high time consumption required to evaluate the performance, a surrogate optimization solver based on \cite{Gutmann2001} and implemented in the MATLAB global optimization toolbox is used.

4D rotation matrices have six degrees of freedom, and hence any such matrix can be parameterized as a product of six Givens-rotation \cite[Sec.~5.1]{golub:mtxcomp} matrices, each with a specific angle. Any order in which the Givens rotations are applied can yield an arbitrary rotation $\mathsf R$, given a suitable set of angles $\{\phi_1,\phi_2,\dots,\phi_6\}$. A possible construction is
\begin{equation}
	\mathsf R=\mathsf G^{34}(\phi_{1})\mathsf G^{12}(\phi_{2})\mathsf G^{24}(\phi_{3})\mathsf G^{23}(\phi_{4})\mathsf G^{14}(\phi_{5})\mathsf G^{13}(\phi_{6}),
	\label{eq:rot_decomp}
\end{equation}
where $\mathsf G^{ik}(\phi)$ denotes a Givens rotation of an angle $\phi$ in dimensions $i$ and $k$ for $1\leq i<k\leq 4$. Note that dimensions 1--2 and 3--4 correspond to rotations of a complex point in the first and second channel, respectively. Since $\mathsf G^{12}(\cdot)$ and $\mathsf G^{34}(\cdot)$ are the last rotations in the matrix product forming $\mathsf R$ in \eqref{eq:rot_decomp}, they simply apply phase shifts to the 2D constellation projections that are transmitted in each channel. It is found that they do not affect the transmission performance for the model in \eqref{eq:sysmodel}, and hence it suffices to let $\phi_1=\phi_2=0$ in \eqref{eq:rot_decomp} such that $\mathsf G^{12}(0)=\mathsf G^{34}(0)=\mathsf I$ and $\mathsf R=\mathsf G^{24}(\phi_3)\mathsf G^{23}(\phi)\mathsf G^{14}(\phi_5)\mathsf G^{13}(\phi_6)$, requiring only four angles to optimize.

\begin{figure}[!t]
	\centering
	\includegraphics{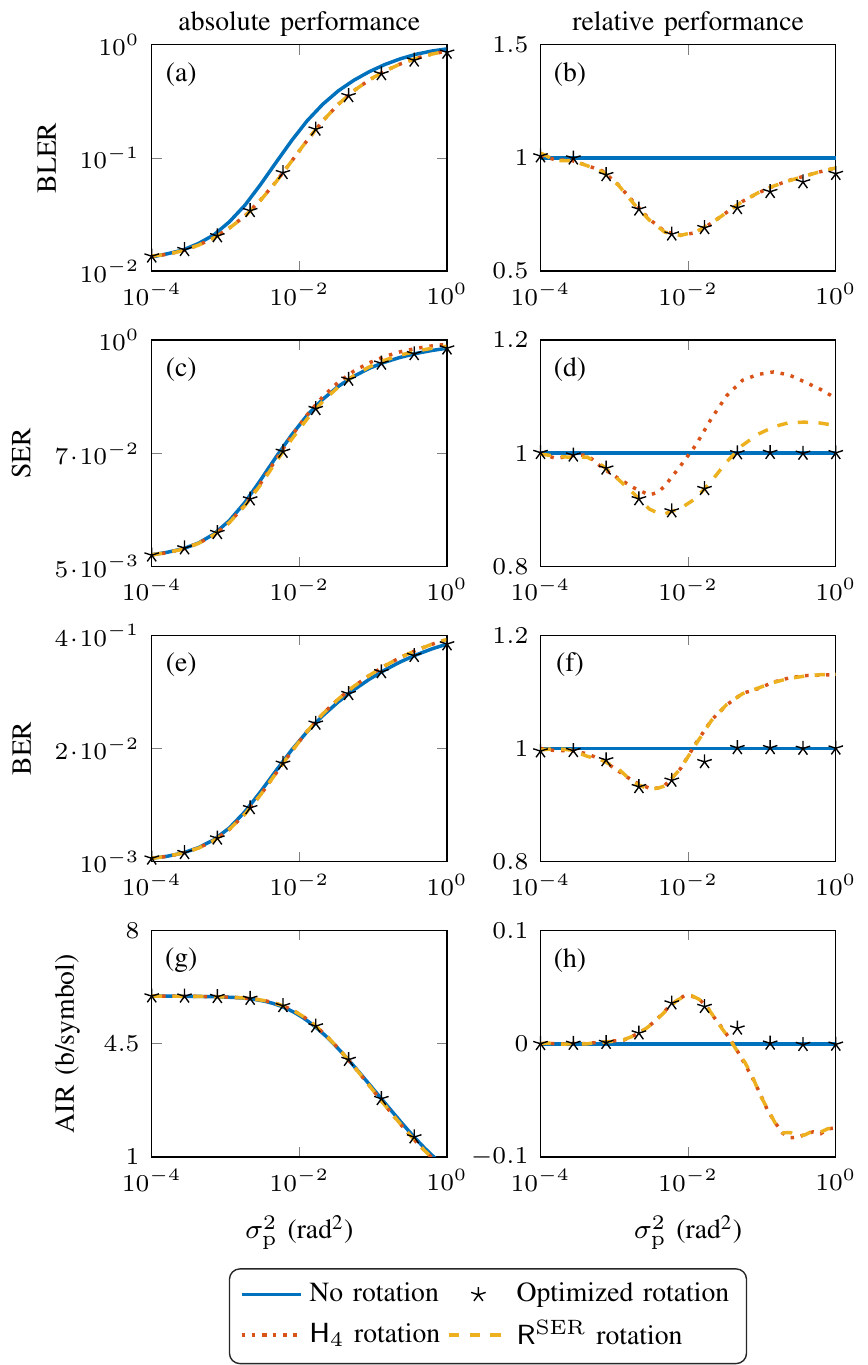}
	\caption{(a)--(b) Absolute and relative BLER performance versus $\rpnvar$ for the joint-channel receiver. (c)--(h)  Absolute and relative performance in terms of BER, SER, and AIR versus $\rpnvar$ for the per-channel receiver.}\vspace{-0.3cm}
	\label{fig:PerformanceVsPNvar}
\end{figure}

\subsection{Results}

Fig.~\ref{fig:PerformanceVsPNvar}(a) shows BLER results for transmission of rotated Gray-coded 64QAM at an SNR of 22.5 dB as a function of the residual phase-noise variance $\rpnvar$, using the joint-channel receiver for detection. Fig.~\ref{fig:PerformanceVsPNvar}(b) shows the same results but relative to the unrotated case in order to emphasize the performance impact of the considered rotations. Note that the absence of rotations corresponds to a signal that has not been processed at the transmitter. As evident from Fig.~\ref{fig:PerformanceVsPNvar}(b), the residual phase noise has a marginal impact on the transmission for $\rpnvar\leq10^{-3}$ rad\textsuperscript2, in which case the performance is invariant under rotations. This is to be expected since signal rotations have no effect on transmission limited only by AWGN \cite{681321}. As the residual phase noise grows stronger, rotations begin to affect the system performance, and $\mathsf H_4$ sees the same performance as the numerically-optimized rotations until $\rpnvar$ exceeds $10^{-1}$ rad\textsuperscript2. Since $\rpnvar\in[10^{-4},10^{-2}]$ rad\textsuperscript2 corresponds to typical amounts of residual phase noise, it can be concluded that when the joint-channel receiver is used, Hadamard rotations are near-optimal in terms of BLER for practical purposes. Moreover, even though the results for $\rpnvar\geq10^{-2}$ do not correspond to practical cases involving residual laser phase noise, they do provide an insight into what occurs in extreme scenarios. In particular, they indicate what may be expected in the presence of nonlinear phase noise, which typically has a higher impact on the transmission performance than laser phase noise.

Figs.~\ref{fig:PerformanceVsPNvar}(c)--(h) show analogous results for the per-channel receiver. The relative AIR is defined as the AIR of rotated-signal transmission subtracted by the AIR of unrotated-signal transmission. When $\rpnvar\approx10^{-4}$ rad\textsuperscript2, rotations have negligible effect on the performance as the AWGN is the main limitation. For $\rpnvar\in[10^{-4},10^{-2}]$ rad\textsuperscript2, $\mathsf H_4$ attains the same performance as the numerically-optimized rotations in terms of BER and AIR. However, for $\rpnvar>10^{-2}$ rad\textsuperscript2, signal rotations become detrimental to the performance.
To gain intuition for this fact, consider transmission of rotated constellations and the use of the per-channel receiver in Section \ref{sec:pc_rec}. After derotation of $\rndvec{r}$, the residual phase noise in each channel manifests as signal attenuation and interference from other channels\footnote{This is detailed for Hadamard rotations in Appendix \ref{sec:hadanalysis}, but similar results can also be found for OFDM transmission in the presence of phase noise \cite{948268}.}. For certain values of $\rpnvar$, this becomes more harmful than pure residual phase noise when the per-channel receiver is used.

\begin{figure}[!t]
	\centering
	\includegraphics{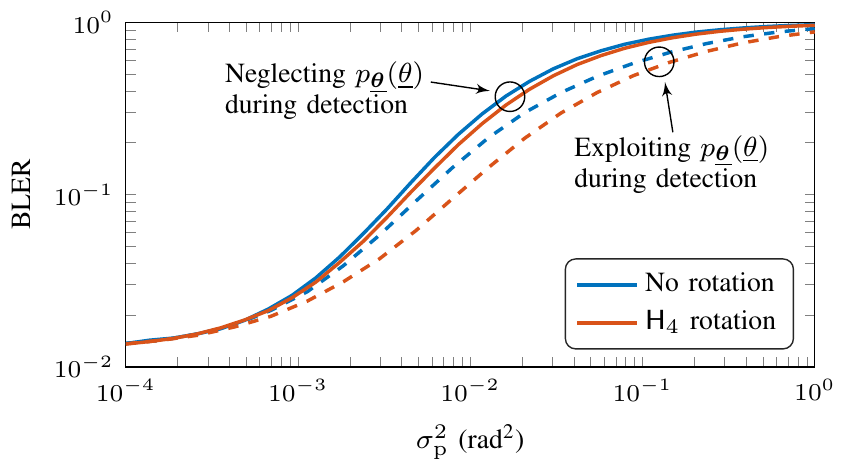}
	\vspace{-0.3cm}
	\caption{BLER versus $\sigma_\mathrm{p}^2$, quantifying the contributions of rotations and $p_{\rndvec{\theta}}(\vec\theta)$ to the detection performance.}\vspace{-0.3cm}
	\label{fig:rot_PNS}
\end{figure}

Interestingly, Figs.~\ref{fig:PerformanceVsPNvar}(c)--(d) show that $\mathsf H_4$ does not yield optimal SER performance, even at moderate amounts of residual phase noise. Instead, it is found that a rotation using
\begin{equation}
	\mathsf R^\mathrm{SER}
	=\frac{1}{\sqrt{2}}
	\begin{bmatrix}
		1 & 1 & 0 & 0\\
		0 & 0 & 1 & 1\\
		1 & -1 & 0 & 0\\
		0 & 0 & -1 & 1
	\end{bmatrix}
\end{equation}
is near-optimal until $\rpnvar$ exceeds $4\cdot10^{-2}$ rad\textsuperscript2, after which rotations become detrimental.
Although $\mathsf R^\mathrm{SER}$ outperforms $\mathsf H_4$ in terms of SER, the two rotations perform identically with respect to BLER, BER, and AIR. This suggests that Gray mapping is suboptimal for the considered system model, but the investigation of suitable bit-to-symbol mappings is deemed out of scope for this paper and is left for future work.

Fig.~\ref{fig:rot_PNS} shows BLER results for transmission of rotated Gray-coded 64QAM at an SNR of 22.5 dB versus $\sigma_\mathrm{p}^2$. The results quantify the performance gains that are achieved in Figs.~\ref{fig:PerformanceVsPNvar}(a)--(b) through effective rotations and symbol detection that utilizes the phase-noise statistics encapsulated in $p_{\rndvec{\theta}}(\vec\theta)$. In particular, the results show that the use of $p_{\rndvec{\theta}}(\vec\theta)$ for detecting unrotated constellations can be more beneficial than detecting rotated constellations and neglecting $p_{\rndvec{\theta}}(\vec\theta)$. However, the former strategy requires the receiver to have knowledge of $p_{\rndvec{\theta}}(\vec\theta)$ when detection takes place. Moreover, these strategies can also be expected to differ in terms of complexity, whose analysis is left for future work.

\begin{figure}[!t]
	\centering
	\includegraphics{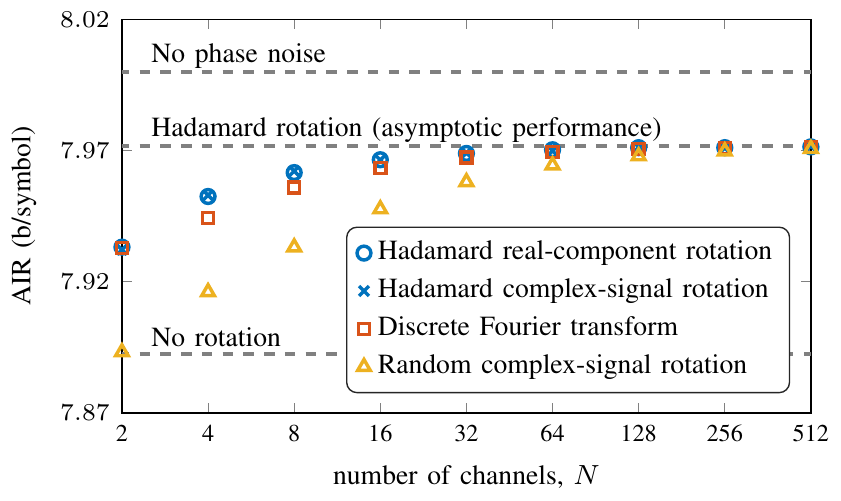}
	\vspace{-0.2cm}
	\caption{AIR for different rotations versus the number of channels for transmission of 256QAM at an SNR of 34 dB with $\rpnvar=10^{-3}$ rad\textsuperscript2.}
	\label{fig:GMIvsNch}\vspace{-0.3cm}
\end{figure}

The results presented so far indicate that Hadamard rotations are near-optimal with respect to BLER, BER, and AIR for practical values of $\rpnvar$ in Gray-mapped 64QAM transmission. The performance improvements are greatest for the joint-channel receiver with up to 35\% decrease in BLER, whereas for the per-channel receiver, performance improvements of up to 6\%, 7\%, and 0.04 b/symbol are observed in SER, BER, and AIR, respectively. The per-channel receiver is arguably more practical due to the lower required complexity and the fact that it does not require knowledge about $p_{\rndvec{\theta}}(\vec\theta)$. Due to the straightforward construction in \eqref{eq:hadamard_construction}, the rest of this paper studies the effects of Hadamard rotations of Gray-coded QAM constellations as the signal dimensionality is increased when the per-channel receiver is used.

\section{Hadamard-Rotation Performance}
\label{sec:hadamard_perf}
In this section, the transmission performance of Hadamard-rotated signals using the per-channel receiver is assessed for different numbers of channels $\Nch$, SNRs, and QAM orders $|\mathcal X|$. It is shown in Appendix \ref{sec:hadanalysis} that the performance of Hadamard rotations on a complex-signal basis, which can be implemented as $\mathsf H_\Nch\rndvec{s}$, is equal to that of Hadamard rotations on a real-component basis, i.e., $f_{\mathsf H_{2\Nch}}(\rndvec{s})$. This is illustrated in Fig.~\ref{fig:GMIvsNch}, where the AIR performance of signals that are Hadamard rotated on a complex-signal and real-component basis is plotted as a function of the number of channels. For comparison, Fig.~\ref{fig:GMIvsNch} also includes the performance of the discrete Fourier transform as well as the average performance of an ensemble of random rotations\footnote{The random rotation matrices are obtained by generating random orthogonal matrices \cite[p.~597]{mezzadri:2007} and retaining those with determinant $+1$.}, carried out on a complex-signal basis. This illustrates the superiority of Hadamard rotations, but as the number of channels grows large, the performance variation between the different rotations becomes small. Furthermore, Fig.~\ref{fig:GMIvsNch} shows the performance for transmission of unrotated signals in addition to the performance for transmission in the absence of phase noise.

The asymptotic behavior of Hadamard rotations in the limit of an infinite number of channels is given by the following proposition.
\begin{prop}
	Consider \eqref{eq:sysmodel} when $N\rightarrow\infty$. Then, the transmission of a Hadamard-rotated signal followed by a corresponding receiver-side derotation yields a signal expressed in each channel as
	\begin{equation}
		\rnd{\tilde r}_{i}=\alpha \rnd s_{i}+\rnd{\tilde n}_{i},
		\label{eq:hadamard_rot_asymp}
	\end{equation}
	for $i=1,\dots,\Nch$, where $\alpha=e^{-\rpnvar/2}$ and $\rnd{\tilde n}_i$ is a complex Gaussian random variable with variance $N_0+E_\mathrm{s}(1-e^{-\rpnvar})$. If $\mathcal X$ is a standard QAM format, $\rnd{\tilde n}_i$ is circularly symmetric.
	\label{prop:asympt}
\end{prop}
\begin{proof}
	See Appendix \ref{sec:hadanalysis}.
\end{proof}
According to Proposition \ref{prop:asympt}, the system model is effectively transformed into an AWGN channel with deterministic signal attenuation when the number of channels grows large. The performance of \eqref{eq:hadamard_rot_asymp} yields an asymptote that is marked in Fig.~\ref{fig:GMIvsNch}.
It can be seen that in order to benefit from Hadamard rotations of QAM constellations, $N=16$ suffices for practical purposes. An interesting question pertains to whether the asymptote resulting from Proposition \ref{prop:asympt} applies also to other classes of orthogonal and unitary transforms when $N\rightarrow\infty$, but addressing this question is left for future work.

Fig.~\ref{fig:signalHist}
\begin{figure}[!t]
	\centering
	\includegraphics{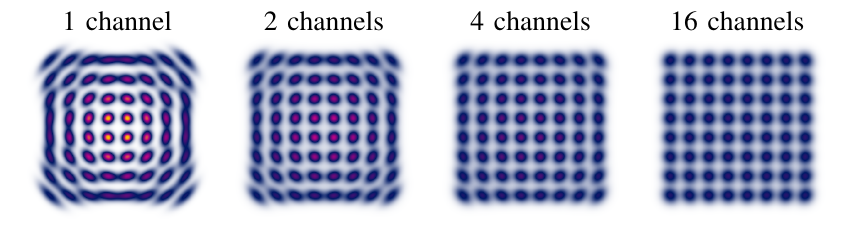}
	\vspace{-0.3cm}
	\caption{The empirical PDF of the received signal in each channel after transmission of Hadamard-rotated 64QAM and receiver-side derotation.}
	\label{fig:signalHist}
\end{figure}
depicts the empirical PDF of a received rotated-64QAM constellation in each channel after derotation at an SNR of 22.5 dB for $\rpnvar=10^{-2}$ rad\textsuperscript2 and different numbers of channels, with 1-channel transmission corresponding to an unrotated signal. As the number of channels increases, the received signal gradually begins resembling the asymptotic case in \eqref{eq:hadamard_rot_asymp}.

\begin{figure}[!t]
	\centering
	\includegraphics{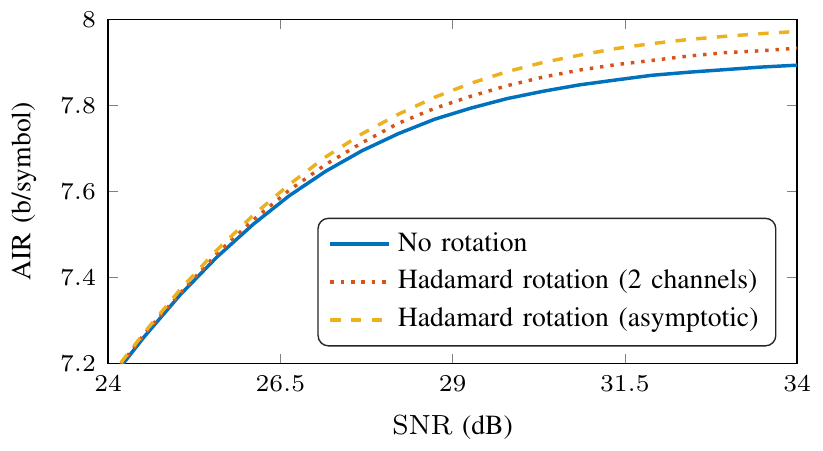}
	\vspace{-0.2cm}
	\caption{AIR versus SNR for transmission of Gray-mapped 256QAM with $\rpnvar=10^{-3}$ rad\textsuperscript2.}
	\label{fig:GMIvsSNR}
\end{figure}

Fig.~\ref{fig:GMIvsSNR} shows the AIR of rotated-256QAM transmission through 2 channels as well as the asymptotic performance as the number of channels grows large, as a function of SNR with $\rpnvar=10^{-3}$ rad\textsuperscript2. As expected, the effect of signal rotations is marginal at lower SNRs since the transmission performance is dominated by AWGN, but as the SNR increases, the impact of rotations becomes significant. At an SNR of 34 dB, 0.04 b/symbol and 0.08 b/symbol improvement in AIR is observed for rotated-256QAM transmission through 2 channels and in the asymptotic case, respectively.

\begin{figure}[!t]
	\centering
	\includegraphics{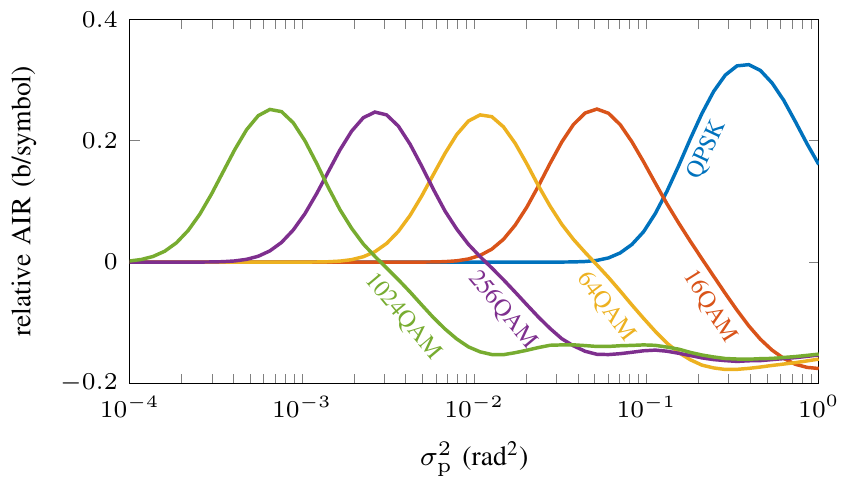}
	\vspace{-0.2cm}
	\caption{The relative AIR for transmission of different Hadamard-rotated QAM constellations versus $\rpnvar$ when the number of channels and SNR grow large.}
	\label{fig:GMIvsPNvar}
\end{figure}

It is clear from the results presented so far that the impact of rotations increases with the number of channels and the SNR. Thus, in order to determine the maximum impact of Hadamard rotations on QAM transmission, Fig.~\ref{fig:GMIvsPNvar} shows the AIR of rotated-QAM transmission relative to the AIR of unrotated-QAM transmission, and versus $\rpnvar$ in the asymptotic case where the number of channels and SNR grow large. Overall, the maximum increase in AIR is found to be approximately 0.33 b/symbol for QPSK and 0.25 b/symbol for higher-order QAM constellations. The point at which the maximum performance increase occurs depends on the QAM order. Therefore, given the fact that practical values of $\rpnvar$ lie in the range $[10^{-4},10^{-2}]$ rad\textsuperscript2, it can be deduced that rotating QAM constellations is most beneficial for 64QAM and higher.

\section{Conclusion}
\label{sec:concl}

This paper investigated the impact of transmitter-side multidimensional signal rotations on multichannel fiber-optic transmission in the presence of residual laser phase noise. The considered system model was based on a previously proposed multichannel phase-noise model that was experimentally verified for SDM transmission through a weakly-coupled MCF. The model entails transmission of arbitrarily rotated QAM constellations in the presence of i.i.d. Gaussian distributed phase noise (see Figs.~\ref{fig:PEdist} and \ref{fig:typicalRPNvar}). Data detection was carried out using the proposed joint-channel or per-channel receivers.

Through numerical optimization for two-channel transmission, Hadamard rotations were found to be near-optimal in terms of BLER, BER, and AIR for practical amounts of residual phase noise (see Fig.~\ref{fig:PerformanceVsPNvar}). In the case of per-channel detection, it was shown that the performance of Hadamard rotations approaches an asymptote as the number of dimensions over which the rotations are performed grows large (see Figs.~\ref{fig:GMIvsNch} and \ref{fig:signalHist}). The performance impact of rotations also depends on the SNR (see Figs.~\ref{fig:GMIvsSNR} and \ref{fig:GMIvsPNvar}).

While transmitter-based processing has already been shown to mitigate impairments such as PDL and MDL, the results in this work show that the tolerance of a multichannel system to practical amounts of laser phase noise can also be improved through joint-channel processing at the transmitter. In particular, these results apply in the case that the laser phase noise is uncorrelated across the channels. The performance improvement is achieved by transmitting rotated multidimensional higher-order QAM constellations (64QAM and higher) at the cost of complexity. In the case of Hadamard rotations, the associated complexity is on the same order as the fast Fourier transform.
Hence, it is clear that Hadamard rotations are interesting for higher-order modulations used for short-to-medium transmission distances. Furthermore, the paper presents several open questions paving the way to future studies on, e.g., effective bit labelings, the effectiveness of rotations for shaped constellations, the performance difference between joint- and per-channel receivers for different modulation formats, and the benefits of rotations in the presence of nonlinearities.

\appendices
\section{Derivation of Joint-Channel Receiver}
\label{app:jointrec}

Let
\begin{equation}
	\mathcal T_\phi(\kappa)\triangleq\frac{1}{2\pi I_0(|\kappa|)}\exp\left(\Re\left\{\kappa e^{-j\phi}\right\}\right)
	\label{eq:def_tik}
\end{equation}
denote a Tikhonov PDF, where $I_0(\cdot)$ is the modified Bessel function of the first kind and zeroth order, $\kappa\in\mathbb C$, and $\phi\in[-\pi,\pi)$,
\begin{equation}
	\mathcal N_x(\mu,\sigma^2)\triangleq\frac{1}{\sqrt{2\pi\sigma^2}}\exp\left(-\frac{1}{2\sigma^2}(x-\mu)^2\right)
	\label{eq:def_gauss}
\end{equation}
denote a real Gaussian PDF, where $x,\mu,\sigma^2\in\mathbb R$, and
\begin{equation}
	\mathcal{CN}\!_z(\mu_z,\sigma^2)\triangleq\frac{1}{\pi\sigma^2}\exp\left(-\frac{1}{\sigma^2}|z-\mu_z|^2\right)
	\label{eq:def_cgauss}
\end{equation}
denote a complex Gaussian PDF, where $z,\mu_z\in\mathbb C$.
The \textit{a posteriori} PMF of $\tilde{\rndvec{s}}$ at $\vec{\tilde s}$ in \eqref{eq:map} can be approximated as
\begin{align}
	P_{\rndvec{\tilde s}|\rndvec{r}}(\vec{\tilde s}|\vec r)&=\int_{\mathbb R^\Nch}p_{\rndvec{\tilde s},\rndvec{\theta}|\rndvec{r}}(\vec{\tilde s},\vec\theta|\vec r)\mathrm{d}\vec\theta\nonumber\\
	&\propto\int_{\mathbb R^\Nch}p_{\rndvec{r}|\rndvec{\tilde s},\rndvec{\theta}}(\vec r|\vec{\tilde s},\vec\theta)p_{\rndvec{\theta}}(\vec\theta)\mathrm{d}\vec\theta\label{eq:appA_1}\\
	&=\prod_{i=1}^\Nch\int_{-\infty}^\infty p_{\rnd r_i|\rnd{\tilde s}_i,\rnd\theta_i}(r_i|\tilde s_i,\theta_i)p_{\rnd\theta_i}(\theta_i)\mathrm d\theta_i\label{eq:appA_2}\\
	&=\prod_{i=1}^\Nch\int_{-\infty}^\infty\mathcal{CN}\!_{r_i}(\tilde s_ie^{j\theta_i},N_0)\mathcal N_{\theta_i}(0,\rpnvar)\mathrm{d}\theta_i\nonumber\\
	&\approx\prod_{i=1}^\Nch\int_{-\pi}^\pi\mathcal{CN}\!_{r_i}(\tilde s_ie^{j\theta_i},N_0)\mathcal{T}_{\theta_i}\left(\frac{1}{\rpnvar}\right)\mathrm{d}\theta_i\label{eq:appA_3}\\
	&\propto\prod_{i=1}^\Nch\exp\left(-\frac{|\tilde s_i|^2}{N_0}\right)\int_{-\pi}^\pi\exp\left(\Re\{\eta_i e^{-j\theta_i}\}\right)d\theta_i\nonumber\\
	&\propto\prod_{i=1}^\Nch\sqrt{2\pi}\exp\left(-\frac{|\tilde s_i|^2}{N_0}\right)I_0\left(\left|\eta_i\right|\right)\label{eq:appA_4}\\
	&\approx\prod_{i=1}^\Nch\frac{1}{\sqrt{|\eta_i|}}\exp\left(\left|\eta_i\right|-\frac{|\tilde s_i|^2}{N_0}\right)\label{eq:mapApprox}
\end{align}
where $\eta_i$ is defined in \eqref{eq:eta_i} and $\propto$ denotes proportionality with respect to $\vec{\tilde s}$. Furthermore, \eqref{eq:appA_1} comes from the Bayes' rule as well as the facts that $\rndvec{\tilde s}$ and $\rndvec{\theta}$ are independent and that $\rndvec{\tilde s}$ is uniformly distributed, \eqref{eq:appA_2} is obtained since the elements in $\rndvec{r}$ and $\rndvec{\theta}$ are mutually independent, \eqref{eq:appA_3} uses $\mathcal N_{\theta_i}(0,\rpnvar)\approx\mathcal T_{\theta_i}(1/\rpnvar)$ for small $\rpnvar$, \eqref{eq:appA_4} comes due to the definitions in \eqref{eq:def_tik} and \eqref{eq:def_cgauss}, and \eqref{eq:mapApprox} uses $I_0(x)\approx e^{x}/\sqrt{2\pi x}$ for large $x>0$.
To circumvent the numerical instability of \eqref{eq:mapApprox}, \eqref{eq:map} can be rewritten as $\hat{\vec s}=\argmax_{\vec s\in\mathcal X^\Nch}\log_e P_{\rndvec{\tilde s}|\rndvec{r}}(\vec{\tilde s}=f_{\mathsf R}(\vec s)|\vec r)$, since the logarithm is a monotonically increasing function. Substituting the logarithm of \eqref{eq:mapApprox} in this expression gives \eqref{eq:log_approx_map}.

\section{Hadamard-Rotation Asymptotic Analysis}
\label{sec:hadanalysis}
As mentioned in Section \ref{sec:rot_opt_results}, applying phase shifts to the 2D constellation projections does not affect the performance results. Hence, Hadamard-rotation matrices can be modified without performance loss such that the Hadamard rotation is effectively applied on a complex-signal basis as opposed to real-component basis. Assuming $N$ to be a power of two, this modification is performed as
\begin{equation}
	\left[\prod_{i=1}^\Nch \mathsf G^{(2i-1)(2i)}\left(\frac{\pi}{4}\right)\right]\mathsf H_{2\Nch}=\mathsf H_\Nch\otimes \mathsf I_2,
	\label{eq:H_mod}
\end{equation}
where the equality in \eqref{eq:H_mod} can be verified by direct calculation.
Thus, instead of carrying out $f_{\mathsf H_{2\Nch}}(\rndvec{s})$, the same performance is achieved through $f_{\mathsf H_{\Nch}\otimes \mathsf I_2}(\rndvec{s})$, which can also be calculated more efficiently as $\mathsf H_\Nch\rndvec{s}$.

Consider transmission of $\rndvec{\tilde s}=f_{\mathsf H_{\Nch}\otimes \mathsf I_2}(\rndvec{s})=\mathsf H_\Nch\rndvec{s}$, in which case \eqref{eq:sysmodel} gives $\rndvec{r}=\rndmtx{\Theta}\mathsf H_\Nch\rndvec{s}+\rndvec{n}$. The derotation of $\rndvec{r}$ yields
\begin{equation}
	\rndvec{\tilde r}=\mathsf H_\Nch^T\rndvec{r}=\mathsf H_\Nch^T\rndmtx{\Theta}\mathsf H_\Nch\rndvec{s}+\mathsf H_\Nch^T\rndvec{n},\label{eq:derot_allch}
\end{equation}
where \eqref{eq:derot_allch} can in each channel be written as
\begin{align}
	\rnd{\tilde r}_i&=\sum_{n=1}^\Nch\sum_{l=1}^\Nch \rnd{s}_nh_{l,i}h_{l,n}e^{j\rnd{\theta}_l}+\sum_{k=1}^\Nch h_{k,i}\rnd{n}_{k}\\
	&=\rnd{s}_i\frac{1}{\Nch}\sum_{l=1}^\Nch e^{j\rnd{\theta}_{l}}+\sum_{n\neq i}\sum_{l=1}^\Nch \rnd{s}_nh_{l,i}h_{l,n}e^{j\rnd{\theta}_l}+\rnd{n}'_i\\
	&=\rnd\alpha_\Nch \rnd{s}_i+\rnd{w}_i+\rnd{n}_i',
\end{align}
for $i=1,\dots,\Nch$. Moreover, $h_{i,j}$ is the $(i,j)$th entry of $\mathsf H_\Nch$, $\rnd\alpha_\Nch$ is the sample average of $e^{j\rnd{\theta}_1},\dots,e^{j\rnd{\theta}_\Nch}$, $\rnd{n}'_i$ has the same distribution as $\rnd{n}_i$, and $\rnd{w}_i$ is interchannel interference that is regarded as additive noise by the per-channel receiver.

The law of large numbers gives $\alpha\triangleq\lim_{\Nch\rightarrow\infty}\rnd\alpha_\Nch=\mathbb{E}[e^{j\rnd{\theta}}]$ where $\rnd\theta$ is a zero-mean Gaussian random variable with variance $\rpnvar$. Moreover, $\mathbb{E}[e^{j\rnd{\theta}}]$ is a special case of the characteristic function of a real Gaussian random variable \cite[Sec.~4.7]{miller:probability}, and hence, $\alpha=e^{-\rpnvar/2}$. Furthermore, $\rnd w_i$ can be rewritten as $\sum_{n\neq i}\rnd s_n\rnd t_{i,n}$, where $\rnd t_{i,n}=\sum_{l=1}^\Nch h_{l,i}h_{l,n}e^{j\rnd\theta_l}$.
Since the columns of Hadamard-rotation matrices are mutually orthogonal, the set $\{h_{1,i}h_{1,n},\dots,h_{\Nch,i}h_{\Nch,n}\}$ has an equal number of $1/\Nch$ and $-1/\Nch$ elements for $n\neq i$. Thus, $\rnd t_{i,n}$ is a sum of $\Nch/2$ i.i.d. random variables with the same distribution as $(e^{j\rnd\theta^+}-e^{j\rnd\theta^-})/\Nch$, where  $\rnd\theta^+$ and $\rnd\theta^-$ are independent zero-mean Gaussian random variables with variance $\rpnvar$. Hence,
\begin{equation}
	\mathbb{E}\left[\rnd t_{i,n}\right]=\frac{\Nch/2}{\Nch}\left(\mathbb{E}\left[e^{j\rnd\theta^+}\right]-\mathbb{E}\left[e^{j\rnd\theta^-}\right]\right)=0,
\end{equation}
and the variance and pseudo-variance \cite[Sec.~7.8]{gallager:digcomm} of $\rnd t_{i,n}$ are
\begin{align}
	\mathbb{E}\left[|\rnd t_{i,n}|^2\right]&=\frac{\Nch/2}{\Nch^2}\left(\mathbb{E}\left[|e^{j\rnd\theta^+}|^2\right]-\mathbb{E}\left[e^{j\rnd\theta^+}\right]\mathbb{E}\left[e^{-j\rnd\theta^-}\right]\right.\nonumber\\
	&\hspace{2.2cm}\left.-\mathbb{E}\left[e^{-j\rnd\theta^+}\right]\mathbb{E}\left[e^{j\rnd\theta^-}\right]+\mathbb{E}\left[|e^{j\rnd\theta^-}|^2\right]\right)\nonumber\\
	&=\frac{1}{\Nch}\left(1-e^{-\rpnvar}\right),\label{eq:varZk}\\
	\mathbb{E}\left[\rnd t_n^2\right]&=\frac{\Nch/2}{\Nch^2}\left(\!\mathbb{E}\left[e^{j2\rnd\theta^+}\right]\!-\!2\mathbb{E}\left[e^{j\rnd\theta^+}\right]\mathbb{E}\left[e^{\rnd\theta^-}\right]\!+\!\mathbb{E}\left[e^{j2\rnd\theta^-}\right]\!\right)\nonumber\\
	&=\frac{1}{\Nch}\left(e^{-2\rpnvar}-e^{-\rpnvar}\right)\label{eq:pvarZk}.
\end{align}

Finally, $\mathbb{E}[\rnd w_i]=0$, and the variance of $\rnd w_i$ can thus be expressed as
\begin{equation}
	\mathbb{E}\left[|\rnd w_i|^2\right]=\sum_{n\neq i}\mathbb{E}\left[|\rnd  s_n|^2\right]\mathbb{E}\left[|\rnd  t_{i,n}|^2\right]=\frac{\Nch-1}{\Nch}E_\mathrm{s}\left(1-e^{-\rpnvar}\right).
\end{equation}
Therefore, due to the central limit theorem, $\rnd w_i$ becomes a complex Gaussian random variable with variance $E_\mathrm{s}(1-e^{-\rpnvar})$ when $N\rightarrow\infty$. Furthermore, if $\mathbb{E}[\rnd s_n^2]=0$ for all $n\neq i$, which is the case when standard QAM formats are transmitted in all channels outside the $i$th channel, the pseudo-variance of $\rnd w_i$ will be
\begin{equation}
	\mathbb{E}\left[\rnd w_i^2\right]=\sum_{n\neq i}\mathbb{E}\left[\rnd s_n^2\right]\mathbb{E}\left[\rnd t_{i,n}^2\right]=0,
\end{equation}
i.e., $\rnd w_i$ will be a circularly symmetric complex Gaussian variable, which leads to \eqref{eq:hadamard_rot_asymp}. On the other hand, if $\mathbb{E}[\rnd s_n^2]\neq0$ for some $n$, e.g., when a PAM format is transmitted in some channel outside the $i$th channel, then $\mathbb{E}[\rnd w_i^2]\neq0$, which implies that $\rnd w_i$ and $\rnd{\tilde r}_i$ will not be circularly symmetric.


\end{document}